\documentclass{article}

\usepackage[top=1in, bottom=1in, left=1in, right=1in]{geometry}
\usepackage{setspace}
\usepackage{indentfirst}
\usepackage{amsmath,amsthm,amssymb}
\usepackage{makecell}
\usepackage{graphicx, xcolor, float}
\graphicspath{{./figs/}}
\usepackage{url}
\usepackage{caption}
\usepackage{subcaption}
\usepackage{enumerate}
\usepackage{booktabs}
\usepackage{natbib}
\usepackage{bm}
\usepackage{todonotes}
\usepackage{algorithm}
\usepackage{algpseudocode}
\usepackage{authblk}

\newcommand{\com}[2]{
  \lb\begin{array}{cc}
    #1 \\ #2
  \end{array}\rb
}
\newcommand{\rb}{\right)}
\newcommand{\lb}{\left(}

\newcommand{\E}{\mathbb{E}}

\newcommand{\B}{\mathcal{B}}

\newcommand{\cP}{\mathcal{P}}
\newcommand{\cPb}{\mathcal{P}_{B}}

\newcommand{\F}{\mathbf{F}}
\renewcommand{\L}{\mathcal{L}}
\newcommand{\class}{\mathbf{U}}
\newcommand{\bclass}{\mathbf{U}^0}
\newcommand{\R}{\mathbb{R}}

\newcommand{\estclass}{\mathbf{H}}

\newcommand{\eps}{\epsilon}

\newcommand{\zero}{\textbf{0}}

\newcommand{\td}[1]{\tilde{#1}}
\newcommand{\Var}{\mathrm{Var}}
\newcommand{\Cov}{\mathrm{Cov}}
\newcommand{\tr}{\mathrm{tr}}

\renewcommand{\vec}{\mathrm{vec}}

\newcommand{\OLS}{\mathrm{OLS}}
\newcommand{\G}{\mathbf{G}}
\newcommand{\cont}{\mathrm{cont}}
\newcommand{\disc}{\mathrm{disc}}

\newcommand{\lin}{\mathrm{lin}}
\newcommand{\inde}{\mathrm{indep}}

\newcommand{\GLS}{\mathrm{GLS}}
\newcommand{\LUE}{\mathrm{LUE}}
\newcommand{\LPQ}{\mathrm{LPQ}}
\newcommand{\koopmann}{\mathrm{Km}}
\renewcommand{\sp}{\mathrm{span}}
\newcommand{\iid}{\mathrm{iid}}
\newcommand{\cl}{\mathrm{cl}}
\newcommand{\sym}{\mathrm{sym}}

\newcommand{\rank}{\mathrm{rank}}

\newtheorem{theorem}{Theorem}[section]
\newtheorem{lemma}{Lemma}[section]
\newtheorem{proposition}{Proposition}[section]

\newtheorem{definition}{Definition}[section]

\title{What Estimators Are Unbiased For Linear Models?}

\author[1]{Lihua Lei \thanks{lihualei@stanford.edu}}
\author[2]{Jeffrey Wooldridge \thanks{wooldri1@msu.edu}}
\affil[1]{Graduate School of Business, Stanford University}
\affil[2]{Department of Economics, Michigan State University}
\begin{document}

\maketitle 

\begin{abstract}
  The recent thought-provoking paper by \citet[Econometrica]{hansen2022modern} proved that the Gauss-Markov theorem continues to hold without the requirement that competing estimators are linear in the vector of outcomes. Despite the elegant proof, it was shown by the authors and other researchers that the main result in the earlier version of Hansen's paper does not extend the classic Gauss-Markov theorem because no nonlinear unbiased estimator exists under his conditions. To address the issue, \cite{hansen2022modern} added statements in the latest version with new conditions under which nonlinear unbiased estimators exist. 

  Motivated by the lively discussion, we study a fundamental problem: what estimators are unbiased for a given class of linear models? We first review a line of highly relevant work dating back to the 1960s, which, unfortunately, have not drawn enough attention. Then, we introduce notation that allows us to restate and unify results from earlier work and \cite{hansen2022modern}. The new framework also allows us to highlight differences among previous conclusions. Lastly, we establish new representation theorems for unbiased estimators under different restrictions on the linear model, allowing the coefficients and covariance matrix to take only a finite number of values, the higher moments of the estimator and the dependent variable to exist, and the error distribution to be discrete, absolutely continuous, or dominated by another probability measure. Our results substantially generalize the claims of parallel commentaries on \cite{hansen2022modern} and a remarkable result by \cite{koopmann1982parameterschatzung}.
\end{abstract}


\section{Introduction}

The celebrated Gauss-Markov theorem plays a central role in econometrics theory for fixed-design linear models. It states that the ordinary least squares (OLS) estimator is the best linear unbiased estimator (BLUE) when the variance-covariance has a scalar form. The result was later generalized by \cite{aitkin1935least}, who proved that the generalized least squares (GLS) estimator is BLUE under a linear model with the covariance matrix known up to a multiplicative constant. Notably, both efficiency results are proved within the class of linear estimators, raising the possibility that there might exist unbiased estimators nonlinear in the response variable that outperform the OLS/GLS estimators.

To the best of our knowledge, the problem of finding unbiased, nonlinear estimators with smaller variances was first studied by Theodore W. Anderson in the early 1960s \citep{anderson1962least}, though the work seems to have received little attention. Anderson found that the OLS estimator does not have the minimal variance among all estimators that are unbiased under homoskedastic linear models with independent errors and an essentially unrestricted design matrix. In particular, he pointed out that linear-plus-quadratic (LPQ) estimators can outperform the OLS estimator for certain error distributions.

Very recently, \cite{hansen2022modern} proved that, among all estimators that are unbiased under any linear model with a finite covariance matrix, the OLS estimator achieves the minimal variance under a scalar covariance matrix. In fact, Hansen proved the more general result that the GLS estimator achieves the minimal variance when the covariance matrix is proportional to the weight matrix (Theorem 4 therein). In this sense, the OLS/GLS estimator is the best unbiased estimator (BUE), a stronger property than BLUE. Hansen's proof relies on a clever combination of the distribution-tilting technique, which is often applied in semiparametric statistics to study asymptotic efficiency \citep[e.g.][]{bickel1993efficient}, and the Cram\'{e}r-Rao lower bound, which is used to study finite-sample efficiency (among unbiased estimators). These two results are not contradictory because \cite{anderson1962least} examines a larger class of candidate estimators which are only required to be unbiased under a scalar covariance matrix while \cite{hansen2022modern} restricts the candidates to be unbiased without constraints on the covariance matrix.

The broader model class considered by \cite{hansen2022modern} only restricts the first moment of the dependent variable, casting doubt on existence of nonlinear unbiased estimators under such weak restrictions. In fact, since the acceptance of \cite{hansen2022modern}, there has been a lively online discussion on this question led by the current authors\footnote{See e.g., \url{https://twitter.com/jmwooldridge/status/1492990971218440197?s=21} and

  \url{https://twitter.com/lihua_lei_stat/status/1493291015129550849?s=21}} and investigated by other researchers \citep{potscher2022modern, portnoy2022linear}. They unanimously conclude, with different proof strategies, that only linear estimators are unbiased under every linear model with finite covariance matrix. This, unfortunately, implies that BUE and BLUE are equivalent and hence Hansen's new results are restatements of the classical Gauss-Markov theorem.

In an attempt to address the issue, \cite{hansen2022modern} added new results in the latest version (Theorems 5 and 6). The new results imply that, among all estimators that are unbiased under linear models with independent (rather than uncorrelated) errors, the OLS estimator is BUE under homoskedasticity and independence. This class of unbiased estimators surely includes nonlinear ones. Apparently, this is not the only way to enrich the set of unbiased estimators. In the 1980s, \cite{koopmann1982parameterschatzung} studied the class of linear models with a covariance matrix known up to a multiplicative constant and proved that all unbiased estimators under this class are LPQ estimators with certain restrictions on the coefficients. Building upon this remarkable representation theorem, \cite{gnot1992nonlinear} proved that the OLS/GLS estimator is no longer BUE among these LPQ estimators under the same class of linear models. Again, this result does not contradict \citet[Theorems 5 and 6]{hansen2022modern} because they consider non-nested classes of unbiased estimators and different model classes for which the BUE is proved.

The differences among the above results are subtle --- any statement that the OLS/GLS estimator is or is not BUE hinges on (a) the set of competing estimators and (b) the class of linear models under examination. To address the nuances, we introduce a framework that allows restatement of all existing results \citep{anderson1962least, gnot1992nonlinear, hansen2022modern} without ambiguity. This can help clear up the confusion that appears to be prevalent in the community and pave the way for further technical discussion. 

After that, unlike the other commentaries, we will mainly focus on the characterization of unbiased estimators under different model classes, instead of conditions under which the OLS/GLS estimator is (not) BUE. Our main contributions are new representation theorems for unbiased estimators under various classes of linear models. Our results (Theorem \ref{thm:GLS_master} and Theorem \ref{thm:koopmann_master}) substantially generalize \cite{potscher2022modern}, \cite{portnoy2022linear}, and \cite{koopmann1982parameterschatzung} by allowing the coefficients and covariance matrix to take only a finite number of values, the higher moments of the estimator and the dependent variable to exist, and the error distribution to be discrete, absolutely continuous, or dominated by another probability measure. They are based on (a slightly stronger version of) a generic result in \cite{ruschendorf1987unbiased}, which unifies and generalizes a line of work on unbiasedness and completeness \citep[e.g.][]{halmos1946theory, fraser1954completeness, hoeffding1977more, fisher1982unbiased, koopmann1982parameterschatzung}. Roughly speaking, it states that the class of unbiased estimators of zero under a class of distributions defined by a set of moment conditions are linear combinations of these moments. We believe the result can be applied to study other model classes and provide a user-friendly proof in Appendix \ref{subapp:proof_master} for readers who are less familiar with functional analysis in abstract Banach spaces. Finally, we apply the representation theorem to detect a new case where the OLS/GLS estimator is BUE. It is analogous to Theorem 5 of \cite{hansen2022modern} but closer to the spirit of classical Gauss-Markov theorem that only constrains the moments. 


\section{Notation and literature review without ambiguity}

\subsection{Families of linear models}
We start by formally defining several families of linear models that have been investigated in the literature and will be studied in this paper. We require an intricate notation in order to address the nuances that have caused considerable confusion. 

Let $X\in \R^{n\times k}$ denote the design matrix and $Y\in \R^{n}$ denote the vector of the dependent variable. Throughout the paper we treat $X$ as fixed and assume the measurable space of $Y$ is equipped with the standard Borel $\sigma$-algebra on $R^{n}$. For any $r\ge 1$ and coefficient vector $\beta\in \R^{k}$, we denote by $\F_r(X; \beta)$ the class of distributions of $Y\in \R^{n}$ with mean $X\beta$ and finite $r$-th moment,
\begin{equation}
  \label{eq:F_X_beta}
  \F_r(X; \beta) = \{F: \E_{F}[Y] = X\beta, \,\,\E_{F}[\|Y\|^r] < \infty\},
\end{equation}
where $\|\cdot\|$ denotes the Euclidean norm in $\R^{n}$. The second condition $\E_{F}[\|Y\|^r] < \infty$ becomes redundant when $r = 1$. The standard (fixed-design) linear model can be characterized as
\begin{equation}
  \label{eq:F_X}
  \F_r(X) = \bigcup_{\beta\in \R^{k}}\F_r(X; \beta),
\end{equation}
often with $r = 2$ and sometimes with $r = 1$ \citep[e.g.][]{jensen1979linear}. It is clear that $\F_s(X)\subset \F_r(X)$ if $s \ge r$. 

In the econometrics literature, the covariance matrix is often assumed to be known or estimable. We further define the stratum $\F_r(X; \beta, \Lambda)$ of $\F_r(X; \beta)$ by fixing the covariance matrix $\Lambda$, i.e.,
\begin{equation}
  \label{eq:F_X_beta_Lambda}
  \F_r(X; \beta, \Lambda) = \{F: \E_{F}[Y] = X\beta, \,\,\Cov_{F}[Y] = \Lambda, \,\, \E_{F}[\|Y\|^r] < \infty\}.
\end{equation}
Note that $\F_r(X; \beta, \Lambda) = \F_2(X; \beta, \Lambda)$ for any $r \in [1, 2]$ and, for any $r\ge 2$,
\[\bigcup_{\Lambda\succeq \zero_{k\times k}}\F_r(X; \beta, \Lambda) = \F_r(X; \beta),\]
where $\zero_{k\times k}$ denotes the $k\times k$ zero matrix and $A \succeq B$ (resp. $\succ$) iff $A - B$ is positive semidefinite (resp. positive definite). Using this notation, the linear model considered in Aitken's theorem can be expressed as $\F_2^\Sigma(X)$ where
\begin{equation}
  \label{eq:F_Sigma}
  \F_r^\Sigma(X) \triangleq \bigcup_{\beta\in \R^{k}}\bigcup_{\sigma^2 > 0}\F_r(X; \beta, \sigma^2 \Sigma), \quad r\ge 2.
\end{equation}
In particular, the sets $\F_2^0$ and $\F_2$ defined in \cite{hansen2022modern} are equivalent to $\F_2^{I_{n}}(X)$, where $I_{n}$ denotes the $n\times n$ identity matrix, and $\bigcup_{\Sigma\succeq \zero_{k\times k}}\F_2^\Sigma(X) = \F_2(X)$ in our notation; see Remark 2.1 of \cite{potscher2022modern} for a similar clarification. 

All above classes only restrict the moments of $Y$. Sometimes one is willing to impose further independence assumptions on the error terms $Y - X\beta$. Specifically, let
\begin{equation}
  \label{eq:F_r_ind_X_beta_Sigma}
  F_{r, \inde}(X; \beta, \Sigma) \triangleq F_{r}(X; \beta, \Sigma) \cap \left\{F: \Sigma^{-1/2}(Y - X\beta)\text{ has independent entries}\right\}.
\end{equation}
Similarly, we can further restrict the errors to be identically distributed:
\begin{equation}
  \label{eq:F_r_iid_X_beta_Sigma}
  F_{r, \iid}(X; \beta, \Sigma) \triangleq F_{r}(X; \beta, \Sigma) \cap \left\{F: \Sigma^{-1/2}(Y - X\beta)\text{ has i.i.d. entries}\right\},
\end{equation}
Analogous to \eqref{eq:F_X} and \eqref{eq:F_Sigma}, we can define
\begin{equation}
  \label{eq:F_r_ind}
  F_{r, \inde}^{\Sigma}(X) \triangleq \bigcup_{\beta\in \R^{k}}\bigcup_{\sigma^2 > 0}\F_{r, \inde}(X; \beta, \sigma^2 \Sigma), \quad  F_{r, \iid}^{\Sigma}(X) \triangleq \bigcup_{\beta\in \R^{k}}\bigcup_{\sigma^2 > 0}\F_{r, \iid}(X; \beta, \sigma^2 \Sigma).
\end{equation}
Note that $\F_{2, \inde}^{I_{n}}(X)$ (resp. $\F_{2, \iid}^{I_{n}}(X)$) gives all linear models with homoskedastic and independent errors (resp. with i.i.d. errors). All have been previously studied in the literature \citep[e.g.,][]{anderson1962least, gnot1992nonlinear, hansen2022modern} and will be examined in later sections. 

Sometimes we further restrict the above model classes by intersecting them with a dominating class:
\begin{equation}
  \label{eq:cP_mu}
  \cP(\mu) = \{F: F \text{ is absolutely continuous with respect to } \mu\},
\end{equation}
or a subset of $\cP(\mu)$ that only includes distributions with bounded Radon-Nikodym derivatives:
\begin{equation}
  \label{eq:cPb_mu}
  \cPb(\mu) = \{F\in \cP(\mu): dF / d\mu \text{ is almost surely bounded under }\mu\}.  
\end{equation}
Note that the set of absolutely continuous distributions in $\R^{n}$, denoted by $\cP_{\cont}$, is given by $\cP(\mu)$ with $\mu$ chosen as the Lebesgue measure\footnote{$\cP(\mu)$ would include mixed distributions if $\mu$ is mixed.}. Furthermore, as with \cite{potscher2022modern} and \cite{portnoy2022linear}, we also consider the class of discrete distributions\footnote{All results in the paper continue to hold if we redefine $\cP_{\disc}$ to include discrete distributions with a countable support.}:
\begin{equation}
  \label{eq:cP_disc}
  \cP_{\disc} = \{F: F \text{ assigns probability }1\text{ on a finite subset of }\R^{n}\}.
\end{equation}

\subsection{Unbiased estimators for fixed-design linear models}\label{subsec:unbiased_preliminary}
Unbiased estimators exist only if $\beta$ is identifiable. A sufficient and necessary condition, which will be assumed throughout, is
\[\mathrm{rank}(X) = k \le n,\]
in which case the coefficient $\beta$ can be identified as
\[\beta(F) = (X' X)^{-1}X'\E_{F}[Y].\]
where $'$ denotes the transpose. For notational convenience, we will simply write $\beta(F)$ as $\beta$, though it should be kept in mind that $\beta$ is a functional of the distribution of $Y$.

For any $r\ge 1$ and family of linear models $\F(X)$ with a full-rank $X$, denote by $\class_{r}(\F(X))$ as the class of estimators (i.e., measurable functions of $Y$) that are unbiased for $\beta$ with finite $r$-th moments under every distribution in $\F(X)$, i.e.,
\begin{equation}
  \label{eq:UFX}
  \class_{r}(\F(X)) = \{u: \R^{n}\mapsto \R^{k}: \E_{F}[u(Y)] = \beta, \,\, u\in \L^{r}(F) \text{ for all }F\in \F(X)\},
\end{equation}
where $\L^{r}(F)$ denotes the set of measurable functions that have finite $r$-th moments under $F$. The recent work \citep{hansen2022modern, potscher2022modern, portnoy2022linear} studied the special case $r = 1$. Clearly, $\class_r(\cdot)$ satisfies the following monotonicity property, which will be repeatly invoked.
\begin{proposition}\label{prop:monotonicity}
 For any $r\ge 1$, $\class_{r}(\F(X))\subset \class_{r}(\td{\F}(X))$ if $\td{\F}(X) \subset \F(X)$.
\end{proposition}
In the following, we will consider the class of linear estimators
\begin{equation}
  \label{eq:estclass_lin}
  \estclass_{\lin}^{n, k} = \{u(y) = A'y: A\in \R^{n\times k}\},
\end{equation}
and the class of linear unbiased estimators (LUE):
\begin{equation}
  \label{eq:estclass_GLS}
  \estclass_{\LUE}(X) = \{u(y) = A'y: A\in \R^{n\times k}, \,\, A'X = I_{k}\}.
\end{equation}
It is well-known that all LUE estimators are unbiased.
\begin{proposition}\label{prop:estclass_GLS}
  For any $r\ge 1$, let $\F_r(X)$ be defined in \eqref{eq:F_X}. Then,
\[\estclass_{\LUE}(X)\subset \class_r(\F_r(X)).\]
\end{proposition}
\begin{proof}
  For any $F\in \F(X)$ and $A\in \R^{n\times k}$ satisfying $A'X = I_{k}$,
  \[\E_{F}[A'Y] = A'X \beta = \beta.\]
  By Rosenthal's inequality \citep{rosenthal1970subspaces}, \[F\in \F_{r}(X)\Longrightarrow \E_{F}[\|Y\|^r]<\infty \Longrightarrow \E_{F}[\|A'Y\|^r] < \infty\Longrightarrow A'Y \in \L^{r}(F).\]
  Thus, $u(y) = A'y\in \class_r(\F_r(X))$.
\end{proof}

Most recent discussions about \cite{hansen2022modern} are essentially about when $\class_1(\F_1(X))\subset \estclass_{\lin}^{n, k}$. The following result shows that whenever it is true, the LUE estimators represent all unbiased estimators.

\begin{proposition}\label{prop:estclass_lin_estclass_GLS}
  For any $r\ge 1$ and $\F(X)\subset \F_r(X)$ such that $\sp\{\beta: F\in \F(X)\} = \R^{k}$,
  \[\class_r(\F(X))\subset \estclass_{\lin}^{n, k}\Longrightarrow \class_r(\F(X)) = \estclass_{\LUE}(X).\]
\end{proposition}
\begin{proof}
  For any $A\in \R^{n\times k}$, $u(y) = A'y \in \class_r(\F(X))$ iff for any $F\in \F(X)$, $u\in \L^{r}(F)$ and
  \[\E_{F}[A'Y] = A'X \beta = \beta.\]
  By Rosenthal's inequality \citep{rosenthal1970subspaces}, $A'Y \in \L^{r}(F)$ for any $A$. Since $\{\beta: F\in \F(X)\}$ spans $\R^{k}$, we must have $I_{k} - A'X = \zero_{k\times k}$ and hence $u\in \estclass_{\LUE}(X)$. The proof is then completed by Proposition \ref{prop:estclass_GLS}.
\end{proof}
The main result of \cite{potscher2022modern} \footnote{The authors posted a proof on social media that coincides with the first proof by \cite{potscher2022modern} two weeks before; see \url{https://twitter.com/lihua_lei_stat/status/1493291015129550849?s=21}.} can be paraphrased as follows.
\begin{theorem}[\cite{potscher2022modern}, Theorem 3.4; see also \cite{portnoy2022linear}]\label{thm:potscher}
  \[\class_1(\F_2(X)\cap \cP_{\disc})\subset \estclass_{\lin}^{n, k}.\]
\end{theorem}
Together with Proposition \ref{prop:estclass_lin_estclass_GLS}, Theorem \ref{thm:potscher} implies that only LUE estimators can be unbiased under linear models without further constraints on the covariance matrix, even if the model only involves discrete distributions with finite second moments, i.e.,
\[\class_1(\F_2(X)\cap \cP_{\disc}) = \estclass_{\LUE}(X);\]
see also the footnote 7 of \cite{potscher2022modern}. 

While only linear estimators can be unbiased under $\F_2(X)$, the conclusion fails for $\F_{2, \iid}^{I_{k}}(X)$.
\begin{proposition}[\cite{hansen2022modern}, remark above Theorem 5]
  \[\estclass_{\LUE}(X)\subsetneq \class_1\lb\F_{2,\iid}^{I_{k}}(X)\rb.\]
\end{proposition}
\cite{hansen2022modern} provides a simple example that adds a mean-zero term, $Y_i$ multiplied by the leave-$i$-th-observation-out residual, onto the OLS estimator. 

Theorem \ref{thm:potscher} also fails if further assumptions are imposed on the covariance matrix. Consider the class of LPQ estimators:
\begin{equation}
  \label{eq:estclass_LPQ}
  \estclass_{\LPQ}^{n, k} = \{u(y) = A'y + (y' B_1 y, \ldots, y' B_k y)':  A\in \R^{n\times k}, B_j \in \R_{\sym}^{n\times n}, j = 1, \ldots, k\},
\end{equation}
where $\R_{\sym}^{n\times n}$ is the set of all symmetric matrices in $\R^{n\times n}$. This class was first studied by \cite{anderson1962least} who showed that, in the model class $\F_{2, \iid}^{I_{k}}$, an LPQ estimator with appropriately chosen $B_j$'s can have a smaller variance for estimating a specific linear contrast of $\beta$ than the best LUE estimator, unless $n = k$ or the errors are Gaussian \citep[e.g.][]{brudno1957dispersion, kagan1969}.

Subsequently, \cite{koopmann1982parameterschatzung} proved a remarkable representation theorem that all unbiased estimators under $\F_2^\Sigma(X)$ are given by a subset of $\estclass_{\LPQ}^{n, k}$: \footnote{\cite{koopmann1982parameterschatzung} does not require the symmetry of $B_j$, though he commented in a remark that $B_j$ can be assumed to be symmetric without loss of generality because $y'B_j y = y'(B_j + B_j')y / 2$. }

\begin{align}
  \estclass_{\koopmann}^{\Sigma}(X) = \{&u(y) = A'y + (y' B_1 y, \ldots, y' B_k y)':  \nonumber\\
  & A\in \R^{n\times k}, A'X = I_{k},  B_j \in \R_{\sym}^{n\times n}, \tr(B_j \Sigma) = 0, X' B_j X = \zero_{k\times k}, j = 1, \ldots, k\},    \label{eq:estclass_LPQ}
\end{align}

\begin{theorem}[\cite{koopmann1982parameterschatzung}, Theorem 4.1]\label{thm:koopmann_unbiased}
  For any given $\Sigma\succeq \zero_{k\times k}$, 
  \[\class_1(\F_2^\Sigma(X)) = \estclass_{\koopmann}^{\Sigma}(X),\]
  where $\F_2^\Sigma(X)$ be defined in \eqref{eq:F_Sigma}.
\end{theorem}

\cite{koopmann1982parameterschatzung} only presented a proof sketch for this result in his book (Section 3 and 4, Chapter I). It is relatively easy to prove that $\estclass_{\koopmann}^{\Sigma}(X)\subset \class_1(\F_2^\Sigma(X))$ via simple linear algebra. The other direction is much more challenging and we will generalize Theorem \ref{thm:koopmann_unbiased} in Section \ref{sec:main} with a self-contained proof.

\begin{proof}[Proof of $\estclass_{\koopmann}^{\Sigma}(X)\subset \class_1(\F_2^\Sigma(X))$:]
  For any $u\in \estclass_{\koopmann}^{\Sigma}(X)$ with parameters $A, B_1, \ldots, B_k$ and $F\in \class_1(\F_2^\Sigma(X))$, $u\in \L^2(F)$ and
  \[\E_{F}[u(Y)] = \E_{F}[A'Y] + (\E_{F}[Y'B_1 Y], \ldots, \E_{F}[Y'B_k Y]).\]
  Let $\eps = Y - X\beta$. By definition of $\class_1(\F_2^\Sigma(X))$, we have $\E_{F}[\eps] = \zero_{k}$, where $\zero_{k}$ denotes the $k$-dimensional zero vector, and $\E_{F}[\eps\eps'] = \sigma^2\Sigma$ for some $\sigma^2 > 0$. Then,
  \[\E_{F}[A'Y] = A'X\beta = \beta.\]
  For each $j$,
  \begin{align*}
    \E_{F}[Y'B_j Y]
    &= \E_{F}[\beta'X'B_j X\beta] + 2\E_{F}[\beta' XB_j \eps] + \E_{F}[\eps' B_j \eps]\\
    & \stackrel{(i)}{=}\E_{F}[\eps' B_j \eps] = \E_{F}[\tr(\eps'B_j \eps)]\\
    & \stackrel{(ii)}{=} \E_{F}[\tr(B_j \eps\eps')] = \sigma^2\E_{F}[\tr(B_j\Sigma)] \\
    & \stackrel{(iii)}{=} 0,
  \end{align*}
  where (i) invokes the restriction that $X'B_j X = \zero_{k\times k}$ and $\E_{F}[\eps] = \zero_{k}$, (ii) uses the fact that $\tr(AB) = \tr(BA)$ and $\E_{F}[\eps\eps'] = \sigma^2\Sigma$, and (iii) applies the condition that $\E_{F}[\tr(B_j \Sigma)] = 0$.
\end{proof}

\subsection{Revisiting B(L)UE for fixed-design linear models}


\begin{definition}\label{def:BUE}
  We say an estimator $u^{*}(Y)$ is BUE in a class of unbiased estimators $\estclass(X)$ with respect to a class of distributions $\F(X)\in \F_2(X)$ if
  \[\Var_{F}[u(Y)]\succeq \Var_{F}[u^{*}(Y)], \quad \text{for any }u\in \estclass(X) \text{ and any }F\in \F(X).\]
  When $\estclass(X) = \estclass_{\LUE}(X)$, we say $h$ is BLUE with respect to $\F(X)$.
\end{definition}
Define the 
the GLS estimator associated with a covariance matrix $\Sigma$,
\begin{equation}
  \label{eq:GLS}
  \hat{\beta}_{\GLS}^{\Sigma}(Y) = (X'\Sigma^{-} X)^{-}X'\Sigma^{-} Y,
\end{equation}
where $A^{-}$ denotes the generalized inverse of $A$. Note that the OLS estimator  $\hat{\beta}_{\OLS}(Y) = (X'X)^{-1}X'Y$ is equivalent to $\hat{\beta}_{\GLS}^{I_k}(Y)$.

The confusion often arises when $\estclass(X)$ and $\F(X)$ are not explicitly defined. We summarize the existing results on BUE and non-BUE of the OLS and GLS estimators separately in the following two theorems. By paraphrasing them on an equal footing, we hope to reveal the nuances. 

\begin{theorem}\label{thm:summary_OLS}
  \begin{enumerate}[(a)]
\item (Gauss-Markov theorem) $\hat{\beta}_{\OLS}(Y)$ is BLUE with respect to $\F_2(X)$.
\item $\hat{\beta}_{\OLS}(Y)$ is BUE in $\estclass(X)$ with respect to $\F(X)$ if one of the following conditions hold:
  \begin{enumerate}[(1)]
  \item (\cite{hansen2022modern}, Theorem 2; implied by Cram\'{e}r-Rao theorem) 
    \[\estclass(X) = \class_1(\F(X)), \quad \F(X) = \F_2(X) \cap \{\mathrm{Normal}(X\beta, \sigma^2 I_{n})\text{ for some }\sigma^2 > 0\}.\]
  \item (\cite{anderson1962least}, Proposition 2) 
  \[n = k = \rank(X), \quad \estclass(X) = \class_1(\F_{2,\iid}^{I_{n}}(X)), \quad \F(X) = \F_{2}(X).\]
  In fact, \cite{anderson1962least} proves this result by showing that $\estclass(X) = \{\hat{\beta}_{\OLS}(Y)\}$.
\item (\cite{hansen2022modern}, Theorem 6)
  \[\estclass(X) = \class_1\lb\bigcup_{\Sigma \text{ diagonal}}\F_{2,\inde}^{\Sigma}(X)\rb, \quad \F(X) = \F_{2,\inde}^{I_{n}}(X).\]
\item (\cite{fraser1954completeness}, together with Rao-Blackwell Theorem)
  \[X = (1, 1, \ldots, 1)', \quad \estclass(X) = \class_1(\F_{2, \iid}^{I_{n}}(X)), \quad \F(X) = \F_{2, \iid}^{I_{n}}(X).\]
  It is not implied by Theorem 7 of \cite{hansen2022modern}, itself a special case of his Theorem 6.
  \end{enumerate}
\item (\cite{anderson1962least}, Proposition 5) $\hat{\beta}_{\OLS}(Y)$ is \textbf{not} BUE in $\estclass(X)$ with respect to $\F(X)$ if 
  \[\min_{i = 1,\ldots,n}\rank(X_{-i}) = k < n, \quad \estclass(X) = \class_1(\F_{2,\inde}^{I_{n}}(X)), \quad \F(X) = \F_{2,\inde}^{I_{n}}(X),\]
  where $X_{-i}$ is the design matrix with the $i$-th row removed.
  \end{enumerate}
\end{theorem}

\begin{theorem}\label{thm:summary_GLS}
  \begin{enumerate}[(a)]
\item (\cite{aitkin1935least}) For a given $\Sigma \succ \zero_{k\times k}$, $\hat{\beta}_{\GLS}^{\Sigma}(Y)$ is BLUE with respect to $\F_2^{\Sigma}(X)$.
\item For a given $\Sigma \succ \zero_{k\times k}$, $\hat{\beta}_{\GLS}^{\Sigma}(Y)$ is BUE in $\estclass(X)$ with respect to $\F(X)$ if any of the following conditions hold:
  \begin{enumerate}[(1)]
  \item (\cite{hansen2022modern}, Theorem 4)
    \[\estclass(X) = \class_1(\F_2(X)), \quad \F(X) = \F_2^{\Sigma}(X).\]
    However, \cite{potscher2022modern} and \cite{portnoy2022linear} showed that $\class_1(\F_2(X)) \subset \estclass_{\lin}^{n, k}$. Thus, this is alien to (a).
  \item (\cite{hansen2022modern}, Theorem 5) $\Sigma$ is diagonal, 
    \[\estclass(X) = \class_1\lb\bigcup_{\Sigma \text{ diagonal}}\F_{2,\inde}^{\Sigma}(X)\rb, \quad \F(X) = \F_{2, \inde}^{\Sigma}(X).\]
  \end{enumerate}
\item (\cite{gnot1992nonlinear}) $\hat{\beta}_{\GLS}^{\Sigma}(Y)$ is \textbf{not} BUE in $\estclass(X)$ with respect to $\F(X)$ if
  \[\estclass(X) = \class_1(\F_2^{\Sigma}(X)), \quad \F(X) = \F_2^{\Sigma}(X).\]
\end{enumerate}
\end{theorem}

  Contrasting Theorem \ref{thm:summary_OLS} (b) (3) with Theorem \ref{thm:summary_OLS} (c) or Theorem \ref{thm:summary_GLS} (b) (1) with Theorem \ref{thm:summary_GLS} (c), we can see that the choice of $\estclass(X)$ is crucial --- while the OLS/GLS can be BUE in a smaller class of estimators, they may no longer be optimal in a larger class.

\section{Unbiased estimators for fixed-design linear models}\label{sec:main}
\subsection{A master theorem}
For any class of fixed-design linear models $\F(X)\subset \F_r(X)$, let
\begin{equation}
  \label{eq:U0FX}
  \bclass_r(\F(X)) = \{u: \R^{n}\mapsto \R: \E_{F}[u(Y)] = 0, u\in \L^{r}(F),\,\, \text{ for all }F\in \F(X)\}.
\end{equation}
Since $\hat{\beta}_{\OLS}$ is linear, by Rosenthal's inequality \citep{rosenthal1970subspaces},
\[\E_{F}[\|Y\|^r] < \infty\Longrightarrow \hat{\beta}_{\OLS}\in \L^{r}(F).\]
Thus, 
\[u\in \class_r(\F(X)) \Longleftrightarrow u - \hat{\beta}_{\OLS} \in \{(u_1, \ldots, u_{k})': u_j\in \bclass_r(F(X))\} \triangleq \bclass_{r}(\F(X))^{\otimes k},\]
where $\otimes k$ denotes the $k$-th order tensor product, and thus,
\begin{equation}
  \label{eq:UFX_U0FX}
  \class_{r}(\F(X)) = \hat{\beta}_{\OLS} + \bclass_{r}(\F(X))^{\otimes k}.
\end{equation}
As a consequence, to find all unbiased estimators of $\beta$, it suffices to characterize all unbiased estimators of $0$. 

It is not hard to see that both $\F_{r}(X; \beta)$ and $\F_{r}(X; \beta, \Lambda)$ can be formulated as
\begin{equation}
  \label{eq:F_r_G}
  \{F: \E_{F}[g(Y)] = 0,\,\, g(Y)\in \L^{r}(F), \,\, \text{ for all }g\in \G\},
\end{equation}
where
\begin{equation}
  \label{eq:GforF}
  \G = \left\{
    \begin{array}{ll}
      \left\{y\mapsto y_{i} - x_{i}'\beta: i = 1, \ldots, n\right\} & \text{for }\F_{r}(X; \beta)\\
      \left\{y\mapsto y_{i} - x_{i}'\beta: i = 1, \ldots, n\right\} & \\
      \,\,\cup \left\{y\mapsto (y_{i} - x_{i}'\beta)(y_{\ell} - x_{\ell}'\beta) - \Lambda_{i\ell}: i,\ell = 1, \ldots, n\right\} & \text{for }\F_{r}(X; \beta, \Lambda)\\
    \end{array}\right.,
\end{equation}
and $x_i'$ denotes the $i$-th row of $X$.

  Intuitively, if we treat $F$ as an element in the ``dual space'' of real-valued measurable functions, though this is incorrect in the rigorous sense, $\F(X)$ is the orthogonal complement of $\G$ in the ``dual space'' with respect to the bilinear form $(u, F)\mapsto \E_{F}[u(Y)]$. Similarly, $\bclass_r(\F(X))$ can be regarded as the orthogonal complement of $\F(X)$. Then we should expect, in an unrigorous sense, that $\bclass_{r}(\F(X)) = (\G^{\perp})^{\perp}$, which is the closure of the span of $\G$ in an appropriately defined normed vector space. This intuitive argument, together with \eqref{eq:GforF}, suggests that all unbiased estimators under $\F_r(X; \beta)$ are linear and those under $\F_r(X; \beta, \Lambda)$ are LPQ, hence implying the results by \cite{potscher2022modern}, \cite{portnoy2022linear}, and \cite{koopmann1982parameterschatzung}.

  We formalize the above heuristics in the following theorem. Below, for a given class of distributions $\F$, we write $\L^{r}(\F)$ for $\bigcap_{F\in \F}\L^{r}(F)$. Moreover, for any subsets of functions $\mathcal{G}$ and $\mathcal{H}$, we say $\mathcal{G} \subset \mathcal{H}$ almost surely under $F$ if for each $g\in \mathcal{G}$ there exists $h\in \mathcal{H}$ such that $g - h = 0$ almost everywhere under $F$, and say $\mathcal{G} = \mathcal{H}$ if $\mathcal{G}\subset \mathcal{H}$ and $\mathcal{H}\subset \mathcal{G}$ almost surely under $F$.

\begin{theorem}[slight generalization of Theorem 1 of \cite{ruschendorf1987unbiased}]\label{thm:master}
  Let $\G$ be a set of real-valued measurable functions on $\R^{n}$ (equipped with the Borel $\sigma$-field) and
  \[\F_{r, \G} = \left\{F: \E_{F}[g(Y)] = 0, g(Y)\in \L^{r}(F), \,\, \text{for all }g\in \G\right\}.\]
  Then for any $F\in \F_{r, \G}$,
  \[\bclass_{r}(\F_{r, \G}\cap \cP(F)) \subset \cl_{F}(\sp(\G))\cap \L^{r}(\F_{r, \G}\cap \cP(F)) \subset \bclass_{r}(\F_{r, \G}\cap \cPb(F))\text{ almost surely under }F,\]
  where $\cP(F)$ and $\cPb(F)$ are the dominating classes defined in \eqref{eq:cP_mu} and \eqref{eq:cPb_mu}, respectively, and $\cl_{F}(\cdot)$ denotes the closure in $\L^{1}(F)$. 
\end{theorem}

When $\G$ only contains a finite number of functions, as in \eqref{eq:GforF}, it is well-known that $\sp(\G)$ is closed (see Proposition \ref{prop:finite_dim}). Then Theorem \ref{thm:master} implies that
\begin{equation}
  \label{eq:G_finite_1}
  \bclass_r(\F_{r, \G}\cap \cP(F)) \subset \sp(\G).
\end{equation}
On the other hand, for any $F \in \F_{r, \G}$,
\[\G\subset \L^{r}(F)\Longrightarrow \sp(\G)\subset \L^{r}(F), \quad F\in \F_{r, \G}\Longrightarrow F\in \F_{r, \sp(\G)}.\]
As a consequence,
\begin{equation}
  \label{eq:G_finite_2}
  \sp(\G)\subset \bclass_{r}(\F_{r, \G})\subset \bclass_{r}(\F_{r, \G}\cap \cP(F)),
\end{equation}
where the second step follows from Proposition \ref{prop:monotonicity}. \eqref{eq:G_finite_1} and \eqref{eq:G_finite_2} imply the following simplified version of Theorem \ref{thm:master}.
\begin{theorem}\label{thm:master_finite}
  If $|\G| < \infty$, then for any $r\ge 1$ and $F\in \F_{r, \G}$,
  \[\bclass_r(\F_{r, \G}\cap \cP(F)) = \sp(\G) \text{ almost surely under }F.\]
\end{theorem}

Theorem \ref{thm:master_finite} generalizes Theorem 1A and 2A of \cite{hoeffding1977more} and Theorem 3.1 of \cite{koopmann1982parameterschatzung}. It will be our key tool to study unbiased estimators for fixed-design linear models. 

\subsection{Unbiased estimators without second moment constraints}

Much of the recent discussion on \cite{hansen2022modern} is about $\F_{2}(X)$, which solely imposes restrictions on the first moment of $Y$. In this section, we will substantially generalize \cite{potscher2022modern} and \cite{portnoy2022linear}, showing that $\class_{r}(\F(X)) = \estclass_{\GLS}(X)$ even when $r > 1$ and $\F(X)$ is much smaller than $\F_{r}(X)$. We start with a surprising result that all unbiased estimators must be linear even if the domain of $\beta$ has only two distinct values with the same direction and the error distribution is restricted into a dominating class.

\begin{lemma}\label{lem:GLS_master}
  For any $r\ge 1$, $\beta^{*}\in \R^{k}$ and $c_1 < c_2\in \R$, let
  \[\F(X) = \bigcup_{\beta\in \{c_j\beta^{*}: j=1,2\}}\F_{r}(X; \beta).\]
  Further, let $F$ be any probability measure in $\F(X)$ such that
  \[\F_{r}(X; \beta)\cap \cP(F) \neq \emptyset, \quad \text{for all }\beta\in \{c_j\beta^{*}: j=1,2\}.\]
  Then
    \[\bclass_{r}(\F(X)\cap \cP(F))\subset \estclass_{\lin}^{n, 1}\,\, \text{almost surely under }F.\]
\end{lemma}
\begin{proof}
  Throughout this proof, we will write $=$ and $\subset$ for equality and subsets almost surely under $F$. By Proposition \ref{prop:monotonicity},
  \[\bclass_{r}(\F(X)\cap \cP(F))\subset \bigcap_{\beta\in \{c_j\beta^{*}: j=1,2\}}\bclass_{r}(\F_{r}(X; \beta)\cap \cP(F)).\]
Without loss of generality, assume that $F\in \F_{r}(X; c_1\beta^{*})$. Taking any $u\in \bclass_{r}(\F(X)\cap \cP(F))$, we must have $u\in \bclass_{r}(\F_{r}(X; c_1\beta^{*})\cap \cP(F))$. By \eqref{eq:GforF} and Theorem \ref{thm:master_finite}, there exists $a_{i}\in \R$ such that
    \[u(y) =\sum_{i=1}^{n}a_{i}(y_{i} - c_1 x_{i}'\beta^{*}).\]
    Since $u\in \bclass_{r}(\F_r(X; c_2\beta^{*})\cap \cP(F))$ and $\F_{r}(X; c_2\beta^{*})\cap \cP(F)\neq \emptyset$, for any $\td{F}\in \F_{r}(X; c_2\beta^{*})\cap \cP(F)$,
    \[0 = \E_{\td{F}}[u(Y)] = (c_{2} - c_{1})\sum_{i=1}^{n}a_{i}x_{i}'\beta^{*} \Longrightarrow \sum_{i=1}^{n}a_{i}x_{i}'\beta^{*} = 0.\]
    Then $u(y)$ can be simplified into $u(y) = \sum_{i=1}^{n}a_{i}y_{i}$. Therefore,
    \[\bclass_{r}(\F(X)\cap \cP(F))\subset \estclass_{\lin}^{n, 1}.\]
\end{proof}

With Lemma \ref{lem:GLS_master}, we can generalize Theorem \ref{thm:potscher} by restricting the values of $\beta$ to a finite set that spans $\R^{k}$ and allowing the error distribution to be absolutely continuous with respect to the Lebesgue measure or any probability measure. 
\begin{theorem}\label{thm:GLS_master}
  Suppose $\B\subset \R^{k}$ includes at least $2k$ vectors
  $\{c_{jm}\beta_{m}^{*}: j = 1, 2, m = 1, \ldots, k\}$ such that $c_{1m} < c_{2m}\in \R$ and 
\[ \sp(\B) = \R^{k}.\]
  For any $r\ge 1$, let
  \[\F(X) = \bigcup_{\beta\in \B}\F_{r}(X; \beta).\]
  \begin{enumerate}[(a)]
  \item   For any probability measure $F\in \F(X)$ such that $\F_r(X; \beta)\cap \cP(F) \neq \emptyset$ for any $\beta\in \B$, 
    \[\class_r(\F(X)\cap \cP(F)) = \estclass_{\GLS}(X) \,\,\text{almost surely under }F.\]
  \item Let $\mu$ be the Lebesgue measure on $\R^{k}$. Then
    \[\class_r(\F(X)\cap \cP_{\cont}) = \estclass_{\GLS}(X) \,\,\text{almost surely under }\mu.\]
  \item Let $\cP_{\disc}$ be defined in \eqref{eq:cP_disc}. Then
    \[\class_r(\F(X)\cap \cP_{\disc}) = \estclass_{\GLS}(X).\]
  \end{enumerate}
\end{theorem}
\begin{proof}
  By Proposition \ref{prop:monotonicity} and Proposition \ref{prop:estclass_GLS}, for any $r\ge 1$,
  \[\estclass_{\GLS}(X)\subset \class_{r}(\F_r(X))\subset \class_{r}(\F(X)\cap \cP)\]
  for any subset $\cP$. Therefore, it remains to prove $\class_{r}(\F(X)\cap \cP)\subset \estclass_{\GLS}(X)$ with $\B = \{c_{jm}\beta_{m}^{*}: j=1,2, m=1,\ldots,k\}$ in each case.
  \begin{enumerate}[(a)]
  \item   Throughout this proof, we will write $=$ and $\subset$ for equality and subsets almost surely under $F$. Without loss of generality, assume that $F\in \F_{r}(X; c_{11}\beta_1^{*})$. By Proposition \ref{prop:monotonicity} and Lemma \ref{lem:GLS_master} with $\beta^{*} = \beta_1^{*}, (c_1, c_2) = (c_{11}, c_{21})$, we have that
    \[\bclass_r(\F(X)\cap \cP(F))\subset \estclass_{\lin}^{n, 1}.\]
    Thus, for any $u\in \bclass_r(\F(X)\cap \cP(F))$, $ u(y) = a'y$ for some $a\in \R^{n}$. 
    For any $j = 1,2$ and $m = 1, \ldots, k$, choose any distribution $\td{F}$ from $\F_r(X; c_{jm}\beta_{m}^{*})\cap \cP(F)$, which is non-empty. Then
    \[0 = \E_{\td{F}}[u(Y)] = a'X(c_{jm}\beta_{m}^{*}).\]
    Since $\sp(\B) = \R^{k}$, we must have $X'a = \zero_{k}$. Thus,
    \[\bclass_r(\F(X)\cap \cP(F))\subset \{u(y) = a'y: X'a = \zero_{k}\}.\]
    By \eqref{eq:UFX_U0FX},
    \begin{align*}
      \class_r(\F(X)\cap \cP(F)) &= \left\{u(y) = \hat{\beta}_{\OLS}(y) + \td{A}'y: \td{A}'X = \zero_{k\times k}\right\}\\
                               & = \left\{u(y) = ((X'X)^{-1}X' + \td{A}')y: \td{A}'X = \zero_{k\times k}\right\}\\
      & \subset \estclass_{\GLS}(X).
    \end{align*}
  \item Let $F \sim N(X(c_{11}\beta_1^{*}), I_{k})$, the multivariate normal distribution with mean $X(c_{11}\beta_1^{*})$ and covariance matrix $I_{k}$. Clearly, $F\in \F(X)$ and, for any $\beta\in \B$, $j=1,2$, and $m = 1,\ldots, k$,
    \[N(X(c_{jm}\beta_{m}^{*}), I_{k})\in \F_{r}(X; c_{jm}\beta_{m}^{*})\cap \cP(F).\]
    Then, the part (a) entails that
    \[\class_r(\F(X)\cap \cP(F))\subset \estclass_{\GLS}(X)\text{ almost surely under }F.\]
    Since $F$ has positive density with respect to $\mu$ everywhere, $\cP(F) = \cP(\mu)$ and a function is almost surely zero under $F$ iff it is so under $\mu$. Therefore,
    \[\class_r(\F(X)\cap \cP_{\cont}) = \class_r(\F(X)\cap \cP(\mu))\subset \estclass_{\GLS}(X)\text{ almost surely under }\mu.\]
  \item The proof is quite technical; see Appendix \ref{subapp:GLS_master}.
  \end{enumerate}
\end{proof}

\subsection{Unbiased estimators with second moment constraints}
Next, we will consider the case where the covariance matrix is constrained. Perhaps surprisingly, even if $\Cov(Y)$ can only take two distinct values that differ by a multiplicative constant and $\beta$ can only take three distinct values pointing to the same direction, any unbiased estimator must  be LPQ. 

\begin{lemma}\label{lem:koopmann_master}
  For any $r\ge 2$, $\beta^{*}\in \R^{k}$, $c_1 < c_2 < c_3\in \R$, $\Sigma \succ \zero_{k\times k}$, and $0 < \sigma_{1}^2 < \sigma_{2}^2$, let
  \[\F(X) = \bigcup_{\beta\in \{c_j\beta^{*}: j=1,2,3\}}\bigcup_{\Lambda \in \{\sigma_{j}^2 \Sigma: j = 1,2\}}\F_{r}(X; \beta, \Lambda).\]
  Further, let $F$ be any probability measure in $\F(X)$ such that
  \[\F_{r}(X; \beta, \Lambda)\cap \cP(F) \neq \emptyset, \quad \text{for all }\beta\in \{c_j\beta^{*}: j=1,2,3\}, \,\Lambda\in \{\sigma_{j}^2 \Sigma: j = 1,2\}.\]
  Then
  \[\bclass_r(\F(X)\cap \cP(F))\subset \estclass_{\LPQ}^{n, 1}\,\, \text{almost surely under }F.\]
\end{lemma}
\begin{proof}
  Throughout this proof, we will write $=$ and $\subset$ for equality and subsets almost surely under $F$. By Proposition \ref{prop:monotonicity},
  \[\bclass_{r}(\F(X))\subset \bigcap_{\beta\in \{c_j\beta^{*}: j=1,2,3\}}\bigcap_{\Lambda \in \{\sigma_{j}^2 \Sigma: j = 1,2\}}\bclass_{r}(\F_{r}(X; \beta, \Lambda)).\]
  Without loss of generality, assume that $F\in \F_{r}(X; c_1\beta^{*}, \sigma_1^2 \Sigma)$. Taking any $u\in \bclass_{r}(\F(X))$, we must have $u\in \bclass_{r}(\F_{r}(X; c_1\beta^{*}, \sigma_{1}^2 \Sigma))$. By \eqref{eq:GforF} and Theorem \ref{thm:master_finite}, there exists $a_{i}, B_{i\ell}\in \R$ such that
  \[u(y) =\sum_{i=1}^{n}a_{i}(y_{i} - c_1 x_{i}'\beta^{*}) + \sum_{i,\ell=1}^{n}B_{i\ell} \left[(y_{i} - c_1 x_{i}'\beta^{*})(y_{\ell} - c_1 x_{\ell}'\beta^{*}) - \sigma_{1}^2\Sigma_{i\ell}\right].\]
  We can enforce $B_{i\ell} = B_{\ell i}$ by replacing both with $(B_{i\ell} + B_{\ell i})/2$. Since $u\in \bclass_{r}(\F(X))$, for any $\td{F} \in \F_{r}(X; c_1\beta^{*}, \sigma_{2}^2 \Sigma)\cap \cP(F)$,
    \[0 = \E_{\td{F}}[u(Y)] = (\sigma_{2}^2 - \sigma_{1}^2)\sum_{j,\ell=1}^{k} B_{j\ell}\Sigma_{j\ell}\Longrightarrow \sum_{j,\ell=1}^{k} B_{j\ell}\Sigma_{j\ell} = 0.\]
    Thus, $u$ can be simplified into
    \[u(y) = \sum_{i=1}^{n}a_{i}(y_{i} - c_1 x_{i}'\beta^{*}) + \sum_{i,\ell=1}^{n}B_{i\ell} (y_{i} - c_1 x_{i}'\beta^{*})(y_{\ell} - c_1 x_{\ell}'\beta^{*}).\]
    Again, since $u\in \bclass_{r}(\F(X))$, for any $j=2,3$, $u\in \bclass_r(\F_{r}(X; c_j\beta^{*}, \sigma_{1}^2 \Sigma)\cap \cP(F))$. Choose $F_j\in \F_{r}(X; c_j\beta^{*}, \sigma_{1}^2 \Sigma)\cap \cP(F)$, which is non-empty. Then
    \[0 = \E_{F_j}[u(Y)] = (c_{j} - c_{1})\sum_{i=1}^{n}a_{i}x_{i}'\beta^{*} + (c_{j} - c_{1})^2\sum_{i,\ell=1}^{k}B_{i\ell}(x_{i}'\beta^{*})(x_{\ell}'\beta^{*})\]
    \[\Longrightarrow 0 = \sum_{i=1}^{n}a_{i}x_{i}'\beta^{*} + (c_{j} - c_{1})\sum_{i,\ell=1}^{n}B_{i\ell}(x_{i}'\beta^{*})(x_{\ell}'\beta^{*}).\]
    Since $c_{2} - c_{1} \neq c_{3} - c_{1}$, we must have
    \[\sum_{i=1}^{n}a_{i}x_{i}'\beta^{*} = \sum_{i,\ell=1}^{n}B_{i\ell}(x_{i}'\beta^{*})(x_{\ell}'\beta^{*}) = 0.\]
    Then $u(y)$ can be further simplified into
    \begin{align*}
      u(y) = \sum_{i=1}^{n}a_{i}y_{i} + \sum_{i,\ell=1}^{n}B_{i\ell}y_{i}y_{\ell} - c_{1}\sum_{i,\ell=1}^{n}B_{i \ell} \left\{(x_i \beta^{*})y_{\ell} + (x_\ell'\beta^{*})y_{i}\right\}.
    \end{align*}
    Therefore,
    \[\bclass_{r}(\F(X))\subset \estclass_{\LPQ}^{n, 1}.\]
\end{proof}

We can then generalize Koopmann's representation theorem (Theorem \ref{thm:koopmann_unbiased}). Throughout the rest of the paper, we will define the vectorization of any symmetric matrix $\Omega\in \R_{\sym}^{k\times k}$ as
\begin{equation}
  \label{eq:vec*}
  \vec^{*}(\Omega) =
      (\Omega_{11} / \sqrt{2},
    \Omega_{12},
    \ldots,
    \Omega_{1k},
    \Omega_{22} / \sqrt{2},
    \Omega_{23},
    \ldots,
    \Omega_{2k},
    \ldots,
    \Omega_{(k-1)(k-1)} / \sqrt{2},
    \Omega_{(k-1)k},
    \Omega_{kk} / \sqrt{2})'.
\end{equation}
By definition, for any $\Omega_1, \Omega_2\in \R_{\sym}^{k\times k}$, $\tr(\Omega_1\Omega_2) = \vec^{*}(\Omega_1)'\vec^{*}(\Omega_2)$.

\begin{theorem}\label{thm:koopmann_master}
  Suppose $\B\subset \R^{k}$ includes 
  $\{c_{jm}\beta_{m}^{*}: j = 1, 2, 3, m = 1, \ldots, M\}$ such that $c_{1m} < c_{2m} < c_{3m}\in \R$ and 
\[\sp(\{\beta_1^{*}, \ldots, \beta_{M}^{*}\}) = \R^{k}, \,\, \sp(\{\beta_1^{*}\beta_1^{*'}, \ldots, \beta_{M}^{*}\beta_{M}^{*'}\}) = \R_{\sym}^{k\times k}.\]
  For any $r\ge 1$, $\Sigma \succ \zero_{k\times k}$, and $0 < \sigma_{1}^2 < \sigma_{2}^2$, let
  \[\F(X) = \bigcup_{\beta\in \B}\bigcup_{\sigma^2 \in \{\sigma_{1}^2, \sigma_{2}^2\}}\F_{2r}(X; \beta, \sigma^2\Sigma).\]  
  \begin{enumerate}[(a)]
  \item   For any probability measure $F\in \F(X)$ such that $\F_r(X; \beta, \sigma^2\Sigma)\cap \cP(F) \neq \emptyset$ for any $\beta\in \B$ and $\sigma^2\in \{\sigma_1^2, \sigma_2^2\}$,
    \[\class_r(\F(X)\cap \cP(F)) = \estclass_{\koopmann}^{\Sigma}(X) \,\,\text{almost surely under }F.\]
  \item Let $\mu$ be the Lebesgue measure on $\R^{k}$. Then
    \[\class_r(\F(X)\cap \cP_{\cont}) = \estclass_{\koopmann}^{\Sigma}(X) \,\,\text{almost surely under }\mu.\]
  \item Let $\cP_{\disc}$ be defined in \eqref{eq:cP_disc}. Then
    \[\class_r(\F(X)\cap \cP_{\disc}) = \estclass_{\koopmann}^{\Sigma}(X).\]
  \end{enumerate}
\end{theorem}
\begin{proof}
  As shown in the partial proof of Theorem \ref{thm:koopmann_unbiased}, for any $r\ge 1$, by Proposition \ref{prop:monotonicity},
  \[\estclass_{\koopmann}^{\Sigma}(X)\subset \class_{1}(\F_{2}(X))\subset \class_{1}(\F_{2r}(X)).\]
  By Rosenthal's inequality \citep{rosenthal1970subspaces}, for any $u\in \estclass_{\LPQ}^{n, k}$ and $F\in \F_{2r}(X)$,
  \[\E_{F}[\|Y\|^r] < \infty\Longrightarrow u\in \L^{r}(\F_{2r}(X)).\]
  Thus,
  \[\estclass_{\koopmann}^{\Sigma}(X)\subset \class_{1}(\F_{2r}(X))\cap \L^{r}(\F_{2r}(X)) = \class_{r}(\F_{2r}(X))\]
  By Proposition \ref{prop:monotonicity} again,  for any subset $\cP$, 
  \[\estclass_{\koopmann}^{\Sigma}(X)\subset \class_{r}(\F(X)\cap \cP).\]
 Therefore, it remains to prove $\class_{r}(\F(X)\cap \cP)\subset \estclass_{\koopmann}^{\Sigma}(X)$ with $\B = \{c_{jm}\beta_{m}^{*}: j=1,2,3, m=1,\ldots,k\}$ in each case. 
  \begin{enumerate}[(a)]
  \item   Throughout this proof, we will write $=$ and $\subset$ for equality and subsets almost surely under $F$. Without loss of generality, assume that $F\in \F_{r}(X; c_{11}\beta_1^{*}, \sigma_{1}^2\Sigma)$. By Proposition \ref{prop:monotonicity} and Lemma \ref{lem:koopmann_master} with $\beta^{*} = \beta_1^{*}, (c_1, c_2, c_3) = (c_{11}, c_{21}, c_{31})$, we have that
    \[\bclass_r(\F(X)\cap \cP(F))\subset \estclass_{\LPQ}^{n, 1}.\]
    Then, for any $u\in \bclass_r(\F(X)\cap \cP(F))$, there exists $a\in \R^{n}$ and $B\in \R_{\sym}^{n\times n}$ such that
    \[u(y) = a'y + y'B y.\]
    For any $\beta\in \B, \sigma^2\in \{\sigma_1^2, \sigma_2^2\}$, take $\td{F} \in \F_r(X; \beta, \Lambda)\cap \cP(F)$, which is non-empty. Then
    \begin{align*}
      0 & = \E_{\td{F}}[u(Y)] = a'\E_{\td{F}}[Y] + \E_{\td{F}}[\tr(Y'B Y)] \\
        & = a'\E_{\td{F}}[Y] + \beta'X'B X\beta + 2\E_{\td{F}}[(Y - X\beta)']B X\beta + \E_{\td{F}}[(Y - X\beta)'B(Y - X\beta)]\\
        & = a'\E_{\td{F}}[Y] + \beta'X'B X\beta  + \E_{\td{F}}[(Y - X\beta)'B(Y - X\beta)]\\
        & = a'\E_{\td{F}}[Y] + \beta'X'B X\beta  + \E_{\td{F}}[\tr(B(Y - X\beta)(Y - X\beta)')]\\
      & = a'X\beta + \beta'X'B X\beta + \tr(B \Sigma)\sigma^2.
    \end{align*}
    Fixing $\beta$ and setting $\sigma^2$ to be $\sigma_{1}^2$ and $\sigma_{2}^{2}$, we obtain that
    \begin{equation}
      \label{eq:Km_trace}
      \tr(B\Sigma)\sigma_1^2 = \tr(B\Sigma)\sigma_2^2\Longrightarrow \tr(B\Sigma) = 0.
    \end{equation}
    Thus, for any $\beta\in \B$,
    \[a'X\beta + \beta'X'B X\beta = 0.\]
    For any $m = 1, \ldots, M$, setting $\beta = c_{jm}\beta_{m}^{*}$ yields that, for $c\in \{c_{1m}, c_{2m}, c_{3m}\}$,
    \[(a'X\beta_{m}^{*})c + (\beta_{m}^{*'}X'BX \beta_{m}^{*})c^2 = 0.\]
    This quadratic equation has three distinct solutions iff
    \[a'X\beta_{m}^{*} = 0, \quad \beta_{m}^{*'}X'BX \beta_{m}^{*} = 0.\]
    Since $\sp(\{\beta_{1}^{*}, \ldots, \beta_{m}^{*}\}) = \R^{k}$, we must have
    \begin{equation}
      \label{eq:Km_linear}
      X'a = \zero_{k}.
    \end{equation}
    Similarly,
    \[\beta_{m}^{*'}X'BX \beta_{m}^{*} = 0 \Longleftrightarrow \tr\left[(X'BX) (\beta_{m}^{*}\beta_{m}^{*'})\right] = 0\Longleftrightarrow \vec^{*}(X'BX)'\vec^{*}(\beta_{m}^{*}\beta_{m}^{*'}) = 0,\]
Since $\sp(\{\beta_{m}^{*}\beta_{m}^{*'}: m = 1,\ldots, M\}) = \R_{\sym}^{k\times k}$,
\[\sp(\{\vec^{*}(\beta_{m}^{*}\beta_{m}^{*'}): m = 1,\ldots, M\}) = \R^{k(k-1)/2}.\]
Thus,
\begin{equation}
  \label{eq:Km_X'BX}
  \vec^{*}(X'B X) = 0\Longrightarrow X'BX = \zero_{k\times k}.
\end{equation}
Piecing \eqref{eq:Km_trace}, \eqref{eq:Km_linear}, and \eqref{eq:Km_X'BX} together, we conclude that
\[\class_r(\F(X)\cap \cP(F)) = \hat{\beta}_{\OLS} +  \bclass_{r}(\F(X))^{\otimes k} \subset \estclass_{\koopmann}^{\Sigma}(X).\]
  \item Same as the proof of Theorem \ref{thm:GLS_master} (b) except that $F$ is chosen as $N(X(c_{11}\beta_1^{*}), \sigma_1^2 \Sigma)$.
  \item The proof is quite technical; see Appendix \ref{subapp:koopmann_master}.
  \end{enumerate}
\end{proof}

As mentioned earlier, \cite{koopmann1982parameterschatzung} only presented a proof sketch for Theorem \ref{thm:koopmann_unbiased}, which is a special case of Theorem \ref{thm:koopmann_master} with $\F(X) = \F_{2}^\Sigma(X)$. Therefore, we confirm the remarkable result by \cite{koopmann1982parameterschatzung} and hence address the concerns raised by \cite{potscher2022modern}.

When $r\ge 1$ and no constraint is imposed on $\beta$, i.e., $\B = \R^{k}$, Theorem \ref{thm:GLS_master} and Theorem \ref{thm:koopmann_master} show that the constraint on the covariance matrix matters because
\[\estclass_{\GLS}(X) = \class_r(\F_{2r}(X)) = \class_r\lb \bigcup_{\Sigma \succeq \zero_{k\times k}}\F_{2r}^\Sigma(X)\rb, \quad \estclass_{\koopmann}^\Sigma(X) = \class_{r}\lb\F_{2r}^\Sigma(X)\rb,\]
where $\F_{2r}^{\Sigma}(X)$ is defined in \eqref{eq:F_Sigma}. The following result shows that the quadratic estimators would be excluded even if the domain of $\beta$ and $\Sigma$ are both finite. 

\begin{theorem}\label{thm:Sigma_constraint}
  Let $\B\subset \R^{k}$ satisfies the condition in Theorem \ref{thm:koopmann_master} and $\mathcal{M}\subset \R_{\sym}^{n\times n}$ includes at least $n(n-1)$ matrices $\{\sigma_{jm}^2\Sigma_m: j=1,2, m = 1,\ldots, n(n-1)/2\}$ such that $0 < \sigma_{1m}^2 < \sigma_{2m}^2$ and 
  \[\sp\lb \mathcal{M}\rb = \R_{\sym}^{n\times n}.\]
Given $r\ge 1$, Let
\[\F(X) = \bigcup_{\beta\in \B}\bigcup_{\Lambda\in \mathcal{M}}\F_{2r}(X; \beta, \Lambda),\]
Then,
\[\class_{r}\lb \F(X) \cap \cP_{\disc}\rb = \estclass_{\GLS}(X).\]
The same conclusion holds if $\cP_{\disc}$ is replaced by $\cP_{\cont}$ and $=$ denotes the almost sure equality.
\end{theorem}
\begin{proof}
  Let $u\in \class_{r}(\F(X)\cap \cP_{\disc})$. By Proposition \ref{prop:monotonicity}, for any $m$, 
  \[u \in \class_{r}(\F(X)\cap \cP_{\disc}) \subset \class_{r}\lb\bigcup_{\beta\in \B}\bigcup_{\sigma^2\in \{\sigma_{1m}^2, \sigma_{2m}^2\}}\F_{2r}(X; \beta, \sigma^2\Sigma_{m})\cap \cP_{\disc}\rb.\]
  Since $\B$ satisfies the condition of Theorem \ref{thm:koopmann_master}, we have that
  \[\class\lb\bigcup_{\beta\in \B}\bigcup_{\sigma^2\in \{\sigma_{1m}^2, \sigma_{2m}^2\}}\F_{2r}(X; \beta, \sigma^2\Sigma_{m})\cap \cP_{\disc}\rb = \estclass_{\koopmann}^{\Sigma_{m}}(X).\]
  Then, there exists $A\in \R^{n\times k}$ and $B_1, \ldots, B_k \in \R_{\sym}^{n\times n}$ such that
  \[u(y) = A'y + (y'B_1 y, \ldots, y'B_k y)\]
  where
  \begin{equation}
    \label{eq:trace_constraint}
    A'X = I_k, \quad \tr(B_j \Sigma_m) = 0, \,\, j = 1, \ldots, k.
  \end{equation}
  For each $j$,
  \[\vec^{*}(B_j)'\vec^{*}(\Sigma_{m}) = 0, \,\, m = 1, \ldots, n(n-1)/2.\]
  Since $\{\Sigma_{m}\}$ spans $\R_{\sym}^{n\times n}$, $\{\vec^{*}(\Sigma_{m})\}$ spans $\R^{n(n-1)/2}$. Therefore,
  \[\vec^{*}(B_j) = 0\Longrightarrow B_j = \zero_{n\times n}.\]
  Therefore,
  \[u(y) = A'y\in \estclass_{\GLS}(X).\]
  This implies $\class_{r}(\F(X)\cap \cP_{\disc}) \subset \estclass_{\GLS}(X)$. By Proposition \ref{prop:monotonicity} and Proposition \ref{prop:estclass_GLS}, we know that $\estclass_{\GLS}(X)\subset \class_{r}(\F_r(X)) \subset \class_{r}(\F(X)\cap \cP_{\disc})$. Thus, the result for $\cP_{\disc}$ is proved. The result for $\cP_{\cont}$ can be proved using the same argument.
\end{proof}


\subsection{A new nontrivial case where OLS/GLS is BUE}
With the powerful representation theorem for $\class_{r}(\F_{2r}^\Sigma(X)\cap \cP)$ discussed in Theorem \ref{thm:koopmann_master}, we can show that the OLS/GLS estimator is BUE in a sufficiently rich class that does not rule out quadratic estimators with respect to a subset of $\F_{2r}^\Sigma(X)$.
\begin{theorem}\label{thm:new_BUE}
  Assume $n \ge \max\{k + 2, 4\}$. Let $\G_3^\Sigma(X)$ includes all distributions of $Y$ such that
  \[\E_{F}[\eps_{i_1}\eps_{i_2}\eps_{i_3}] = 0, \,\, \text{ for any }(i_1, i_2, i_3)\in \{1, \ldots, n\}^3\setminus \{(i, i, i): i \le n\}, \,\, \text{where }\eps = \Sigma^{-1/2}(Y - X\beta).\]
Given any $r\ge 1$ and $\Sigma \succ \zero_{k\times k}$, let $VDV'$ be the eigen-decomposition of $\Sigma$ and
  \[\F(X) = \bigcup_{\Sigma\in \mathcal{C}_{\Sigma}}\F_{2r}^\Sigma(X),\]
  where
  \[\mathcal{C}_{\Sigma} = \{V\td{D}V': \td{D}\in \R^{n\times n} \text{ is diagonal}\}.\]
  Then
  \[\class_{r}(\F(X)\cap \cP)\supsetneq \estclass_{\GLS}(X),\]
  and $\hat{\beta}_{\GLS}^{\Sigma}(Y)$ is BUE in $\class_{r}(\F(X)\cap \cP)$ with respect to $\F_{2r}^{\Sigma}(X)\cap \G_3^{\Sigma}(X)$ for $\cP\in \{\cP_{\cont}, \cP_{\disc}\}$.
\end{theorem}
\begin{proof}
  Without loss of generality, we assume that $\Sigma = I_{k}$; otherwise, we replace $Y$ by $\Sigma^{-1/2}Y = VD^{-1/2}Y$ and $X$ by $\Sigma^{-1/2}X = VD^{-1/2}X$, in which case $\hat{\beta}_{\GLS}^{\Sigma}(Y)$ reduces to $\hat{\beta}_{\OLS}(VD^{-1/2}Y)$ and $\mathcal{C}_{\Sigma}$ reduces to $\mathcal{C}_{I_k}$. For any $u\in \class_r(\F(X))$, using the same argument for \eqref{eq:trace_constraint}, we can show that $u(y) = A'y + (y'B_1 y, \ldots, y'B_k y)$ with
  \begin{equation}
    \label{eq:uy_C_Ik}
    A'X = I_{k}, \,\, X'B_j X = \zero_{k\times k}, \,\, \tr(B_j \Sigma) = 0, j= 1,\ldots, k, \Sigma\in \mathcal{C}_{I_k}.
  \end{equation}
Since $\mathcal{C}_{I_k}$ includes all diagonal matrices, we must have 
 \[B_{j, ii} = 0, \,\, i = 1, \ldots, n, j = 1, \ldots, k;\]
that is, all diagonal elements of $B_{j}$ must be zero. For each $j$, this is equivalent to $n$ linear constraints on $\vec^{*}(B_j)$, i.e.,
\begin{equation}
  \label{eq:B_j_eq1}
  \vec^{*}(e_i e_i')'\vec^{*}(B_j) = 0, \quad i = 1, \ldots, n
\end{equation}
where $e_i$ is the $i$-th canonical basis in $\R^{n}$. The other constraint $X' B_j X = 0$ can be formulated as $k(k-1)/2$ linear constraints on $B_j$, i.e.,
\begin{equation}
  \label{eq:B_j_eq2}
  \td{e}_{p}'X'B_j X \td{e}_{q} = 0\Longleftrightarrow \vec^{*}(Xe_{q}e_{p}'X')'\vec^{*}(B_j) = 0, \quad 1\le p\le q \le n,
\end{equation}
where $\td{e}_{p}$ is the $p$-th canonical basis in $\R^{k}$. Note that \eqref{eq:B_j_eq1} and \eqref{eq:B_j_eq2} give all constraints on $B_j$, which correspond to a homogeneous linear system on $\vec^{*}(B_j)\in \R^{n(n-1)/2}$ with $n + k(k - 1)/2$ equations. Since $k \le n - 2$ and $n\ge 4$,
\[n + \frac{k(k-1)}{2}\le n + \frac{(n - 2)(n - 3)}{2} = \frac{n^2 - 3n + 6}{2} < \frac{n(n - 1)}{2}.\]
Therefore, a non-zero solution always exists. For any solution $B$ and $A\in \R^{n\times k}$ with $A'X = I_{k}$, consider the estimator
\[u_{B}(Y) = A'Y + (Y' B Y, Y' B Y, \ldots, Y' B Y).\]
Then, for any $F\in \F(X)$, then
\begin{align}
  \E_{F}[Y' B Y] &= \beta' X' B X \beta + 2\beta'X' B\E_{F}[Y - X\beta] + \E_{F}[(Y - X\beta)'B(Y - X\beta)]\nonumber\\
                 & = \E_{F}[(Y - X\beta)'B(Y - X\beta)] = \E_{F}[\tr(B(Y - X\beta)(Y - X\beta)')] \nonumber\\
  & = \E_{F}[B\Sigma] = 0,\label{eq:EF_Y'BY}
\end{align}
where the last equality is due to that all diagonal elements of $B$ are zero and $\Sigma\in \mathcal{C}_{I_k}$ is diagonal. Therefore,
\[\E_{F}[u_{B}(Y)] = \E_{F}[A'Y] = A'X\beta = \beta\Longrightarrow u_{B} \in \class_1(\F(X)\cap \cP).\]
Since $\F(X)\subset \F_{2r}(X)$, by Rosenthal's inequality \citep{rosenthal1970subspaces}, $u_{B}\in \L^{r}(\F(X))$. Therefore,
\[\class_r(\F(X)\cap \cP)\supset \estclass_{\GLS}(X)\cup \{u_{B}\}\supsetneq \estclass_{\GLS}(X)\].

\noindent Next, we prove that $\hat{\beta}_{\OLS}(Y)$ is BUE in $\class_r(\F(X))$ with respect to $\F_{2r}^{I_k}(X)\cap \G_{3}^{I_{k}}(X)$. For any $u\in \class_r(\F(X))$, it suffices to show that,
\begin{equation}
  \label{eq:BUE_goal}
  \Cov_{F}(\hat{\beta}_{\OLS}(Y), u(Y) - \hat{\beta}_{\OLS}(Y)) = \zero_{k\times k}, \quad \text{for every }F\in \F_{2r}^{I_{k}}(X)\cap \G_3^{I_{k}}(X).
\end{equation}
By \eqref{eq:uy_C_Ik}, we can write $u(y)$ as
\[u(y) = A'y + (y' B_1 y, \ldots, y' B_k y)',\]
where $A'X = I_{k}$, $X'B_{j}X = \zero_{k\times k}$, and $B_{j, ii} = 0$ for all $i = 1,\ldots, n$. Let $\eps = Y - X\beta$. Recall that we assume that $\Sigma = I_{k}$ without loss of generality. By definition of $\F_{2r}^{I_{k}}(X)$ in \eqref{eq:F_Sigma}, for any $F\in \F(X)$, $\Cov_{F}[\eps] = \sigma^2 I_{k}$ for some $\sigma^2 > 0$, and, analogous to \eqref{eq:EF_Y'BY}, 
\[u(Y) - \E_{F}[u(Y)] = A'\eps + 2(B_1 X\beta, B_2 X\beta, \ldots, B_k X\beta)'\eps + (\eps'B_1 \eps, \ldots, \eps'B_k \eps)'.\]
Then for any $j = 1,\ldots, k$,
\begin{align*}
  &\Cov_{F}(\hat{\beta}_{\OLS}(Y), u(Y) - \hat{\beta}_{\OLS}(Y))\\
  &=   \Cov_{F}((X'X)^{-1}X'\eps, (A' - (X'X)^{-1}X')\eps) + 2\Cov_{F}((X'X)^{-1}X'\eps, (B_1 X\beta, B_2 X\beta, \ldots, B_k X\beta)'\eps)\\
  & \quad + \Cov_{F}((X'X)^{-1}X'\eps, (\eps' B_1 \eps, \ldots, \eps' B_k \eps)').
\end{align*}
Since $A'X = I_{k}$,
\begin{align}
  &\Cov_{F}((X'X)^{-1}X'\eps, (A' - (X'X)^{-1}X')\eps)\nonumber\\
  & = \E_{F}[(X'X)^{-1}X' \eps \eps' (A - X(X'X)^{-1})]\nonumber\\
  & = \sigma^2 (X'X)^{-1}X'(A - X(X'X)^{-1}) = \zero_{k\times k}.\label{eq:cov_term1}
\end{align}
Since $X' B_j X = 0$,
\begin{align}
  & \Cov_{F}((X'X)^{-1}X'\eps, (B_1 X\beta, B_2 X\beta, \ldots, B_k X\beta)'\eps)\nonumber\\
  & = \E_{F}[(X'X)^{-1}X'\eps\eps' (B_1 X\beta, B_2 X\beta, \ldots, B_k X\beta)]\nonumber\\
  & = \sigma^2 \E_{F}[(X'X)^{-1}X'(B_1 X\beta, B_2 X\beta, \ldots, B_k X\beta)]\nonumber\\
  & = \sigma^2 \E_{F}[((X'X)^{-1}(X'B_1 X)\beta, (X'X)^{-1}(X'B_2 X)\beta, \ldots, (X'X)^{-1}(X'B_k X)\beta)] = \zero_{k\times k}.\label{eq:cov_term2}
\end{align}
Consider any $B$ with zero diagonal elements. For any $i = 1, \ldots, n$, let $z_{i}'$ be the $i$-th row of $(X'X)^{-1}X'$. Then
\[\Cov_{F}(z_{i}'\eps, \eps' B \eps) = \sum_{m\le n}\sum_{a\neq b \le n}z_{im}B_{ab}\Cov_{F}(\eps_{m}, \eps_{a}\eps_{b}).\]
Since $F\in \G_3^{I_{k}}(X)$, for any $m, a, b$ with $a\neq b$,
\[\E_{F}[\eps_m\eps_a\eps_b] = 0.\]
Since $\E_{F}[\eps_{m}] = 0$,
\[\Cov_{F}(\eps_{m}, \eps_{a}\eps_{b}) = \E_{F}[\eps_m\eps_a\eps_b] - \E_{F}[\eps_{m}]\E_{F}[\eps_{a}\eps_{b}] = 0.\]
Thus,
\begin{equation}
  \label{eq:cov_term3}
  \Cov_{F}((X'X)^{-1}X'\eps, (\eps' B_1 \eps, \ldots, \eps' B_k \eps)') = \zero_{k\times k}.
\end{equation}
Putting \eqref{eq:cov_term1} - \eqref{eq:cov_term3} together, we prove our goal \eqref{eq:BUE_goal} and thus $\hat{\beta}_{\OLS}(Y)$ is BUE.
\end{proof}

When $\Sigma$ is diagonal, Theorem \ref{thm:new_BUE} is analogous to Theorem 5 of \cite{hansen2022modern} (see also Theorem \ref{thm:summary_GLS} (b) (2) for a restatement in our notation), though we impose restrictions on the third moments instead of independence. It is worth emphasizing that neither result implies the other. 

\bibliography{BUE_GM}
\bibliographystyle{plainnat}

\newpage
\appendix

\section{Proofs}
\subsection{Proof of Theorem \ref{thm:master}}\label{subapp:proof_master}
For notational convenience, write $\F$ for $\F_{r, \G}\cap \cP(F)$. Then $\F_{r, \G}\cap \cPb(F) = \F\cap \cPb(F)$ and
\begin{equation*}
  \bclass_r(\F) = \bclass_1(\F)\cap \L^r(\F), \,\, \bclass_r(\F\cap \cPb(F)) = \bclass_1(\F\cap \cPb(F))\cap \L^r(\F\cap \cPb(F))\supset \bclass_1(\F\cap \cPb(F))\cap \L^r(\F)
\end{equation*}
Thus, it suffices to prove
\begin{equation*}
  \bclass_{1}(\F) \subset \cl_{F}(\sp(\G)) \subset \bclass_{1}(\F\cap \cPb(F)).
\end{equation*}

~\\
\noindent \emph{Proof of $\cl_{F}(\sp(\G))\subset \bclass_{1}(\F\cap \cPb(F))$}: for any $u\in \cl_{F}(\sp(\G))$, $u\in \L^{1}(\F)$ and there exists $u_1, u_2, \ldots$ in $\sp(\G)$ such that $u_{n}\rightarrow u$ in $\L^{1}(F)$ as $n\rightarrow \infty$, i.e.,
\begin{equation}
  \label{eq:un_L1}
  \E_{F}[\|u_{n}(Y) - u(Y)\|] \rightarrow 0.
\end{equation}
Since $u_{n}\in \sp(\G)$, there exists $u_{1n}, \ldots, u_{Mn}\in \G$ and $c_{1n}, \ldots, c_{Mn}\in \R$ for some finite integer $M$ such that
  \[u_{n} = \sum_{i=1}^{M}c_{in}u_{in}.\]
Since $u_{in}\in \G$, for any $\td{F}\in \F(X)$, $\E_{\td{F}}[u_{in}(Y)] = 0$ and, thus,
\[\E_{\td{F}}[u_{n}(Y)] = \sum_{i=1}^{M}c_{in}\E_{\td{F}}[u_{in}(Y)] = 0.\]
When $\td{F}\in \F\cap \cPb(F)$, $d\td{F} / dF$ is bounded almost surely. Then, \eqref{eq:un_L1} implies
\[\E_{\td{F}}[\|u_{n}(Y) - u(Y)\|] = \E_{F}\left[\|u_{n}(Y) - u(Y)\|\frac{d\td{F}}{dF}(Y)\right] \rightarrow 0.\]
As a result,
\[\E_{\td{F}}[u(Y)] = \lim_{n\rightarrow \infty}\E_{\td{F}}[u_{n}(Y)] = 0.\]
Thus, $u\in \bclass_1(\F\cap \cPb(F))$.

~\\
\noindent \emph{Proof of $\bclass_{1}(\F) \subset \cl_{F}(\sp(\G))$}: 
it is well-known that, for any probability measure $F$, $\L^{\infty}(F)$, the space of all almost surely bounded functions under $F$, is the dual space of $\L^{1}(F)$. Denote by $\mathbf{1}$ the constant function that maps any element of $\R^{n}$ to $1$, which is clearly a function in $\L^{\infty}(F)$. Define the orthogonal complement of $\G\cup \{\mathbf{1}\}$ in $\L^{1}(F)$ (in the sense of Proposition \ref{prop:hahn_banach}) as 
  \[\G^{\perp}(F) \triangleq (\G\cup \{\mathbf{1}\})^{\perp} = \left\{h\in \L^{\infty}(F): \E_{F}[h(Y)] = 0, \,\,\E_{F}[h(Y)g(Y)] = 0, \,\, \text{for all }g\in \G \right\}.\]
  For any $h\in \G^{\perp}(F)$, $h$ is almost surely bounded, hence there exists $c > 0$ such that $1 + ch\ge 0$ almost surely under $F$. Let $\td{F}$ be a probability measure with Radon-Nikodym derivative $d\td{F} / dF = 1 + ch$. Then $\int 1d\td{F} = \int (1 + ch)dF = 1$ and thus
  \[\td{F} \in \cP(F).\]
  Moreover, for any $g\in \G$,
  \[\E_{\td{F}}[g(Y)] = \E_{F}[g(Y)(1 + ch(Y))] = \E_{F}[g(Y)] + c\E_{F}[g(Y)h(Y)].\]
  The first term is zero since $F\in \F_{r, \G}$ and the second term is zero since $h\in \G^{\perp}(F)$. This implies $\E_{\td{F}}[g(Y)] = 0$ for any $g\in \G$. Moreover, suppose $|1 + ch|\le B$ for some constant $B$. Then, for any $g\in \G$, 
  \[g\in \L^{r}(F)\Longrightarrow \E_{\td{F}}[\|g(Y)\|^r] = \E_{F}\left[\|g(Y)\|^r \frac{d\td{F}}{dF}(Y)\right]\le B \E_{F}\left[\|g(Y)\|^r\right] < \infty\Longrightarrow g\in \L^{r}(\td{F}).\]
  This entails that
  \[\td{F} \in \F_{r, \G}.\]
  Combing two pieces together, we have proved that
  \[\td{F} \in \F_{r, \G} \cap \cP(F) = \F.\]
  Therefore, for any $u\in \bclass_1(\F)$,
  \[0 = \E_{\td{F}}[u(Y)] = \E_{F}[u(Y)(1 + ch(Y))] = \E_{F}[u(Y)] + c\E_{F}[u(Y)h(Y)] = c\E_{F}[u(Y)h(Y)].\]
  As a result,
  \begin{equation}
    \label{eq:Euh}
    \E_{F}[u(Y)h(Y)] = 0.
  \end{equation}
  Moreover, $u\in \bclass_1(\F)$ implies that
  \[\E_{F}[u(Y)\mathbf{1}] = \E_{F}[u(Y)] = 0.\]  
  Since this holds for any $h\in\G^{\perp}(F)$,  
  \begin{equation}
    \label{eq:u_Gr}
    u \in \lb\G^{\perp}(F)\cup \{\mathbf{1}\}\rb^{\perp}= \lb\G^{\perp}(F)\rb^{\perp} \cap \{\mathbf{1}\}^{\perp},
  \end{equation}
  where the outer $\perp$ denotes the orthogonal complement of the subset $\G^{\perp}(F)\cup \{\mathbf{1}\}$ of $\L^{\infty}(F)$ in $\L^1(F)$ (in the sense of Proposition \ref{prop:hahn_banach}). By Proposition \ref{prop:hahn_banach},
  \[\lb\G^{\perp}(F)\rb^{\perp} = \lb\lb \G\cup \{\mathbf{1}\}\rb^{\perp}\rb^{\perp} = \cl_{F}(\sp(\G\cup \{\mathbf{1}\})).\]
  Since $F\in \F_{r, \G}$, $\G\subset \{\mathbf{1}\}^{\perp}$. Therefore,
  \[\lb\G^{\perp}(F)\rb^{\perp} \cap \{\mathbf{1}\}^{\perp} = \cl_{F}(\sp(\G \cup \{\mathbf{1}\}))\cap \{\mathbf{1}\}^{\perp} = \cl_{F}(\sp(\G)).\]
The result is then proved.

\subsection{Proof of Theorem \ref{thm:GLS_master} (c)}\label{subapp:GLS_master}
    Let $e_{1}, \ldots, e_{n}$ be the canonical basis of $\R^{n}$. For any $\beta\in \B$, let
    \[F_{\beta} = \frac{1}{4n}\sum_{i=1}^{n}(\delta_{e_{i}} + \delta_{-e_{i}} + \delta_{2X\beta - e_{i}} + \delta_{e_{i} - 2X\beta}), \quad \td{F}_{\beta} = \frac{1}{2n}\sum_{i=1}^{n}(\delta_{e_{i}} + \delta_{2X\beta - e_{i}}),\]
    where $\delta_{x}$ denotes the Dirac-measure (point mass) at $x$. Then
    \[\E_{F_{\beta}}[Y] = 0, \quad \E_{\td{F}_{\beta}}[Y] = X\beta.\]
For any $y\in \R^{n}$ and $\beta^{*} = c_{11}\beta_1^{*}$, let
\begin{equation*}
  F_{y} = \frac{1}{4}(\delta_{y} + \delta_{4X\beta^{*} - y}) + \frac{1}{2|\B|}\sum_{\beta\in \B}F_{\beta}.
\end{equation*}
    Then
    \[\E_{F_{y}}[Y] = \frac{1}{4}\cdot 4X\beta^{*} + \frac{1}{2}\cdot 0 = X\beta^{*}.\]
    As a result, $F_{y}\in \F(X)\cap \cP_{\disc}$ and $\td{F}_{\beta}\in \cP(F_{\beta}) \in \cP(F_{y})$ for any $\beta\in \B$. Thus, 
    \[\F_{r}(X; \beta)\cap \cP_{\disc}\supset \F_{r}(X; \beta)\cap \cP(F_{y})\supset \{\td{F}_{\beta}\} \neq\emptyset, \quad \text{for any }\beta\in \B.\]
    By Proposition \ref{prop:monotonicity} and the part (a),
    \[\class_{r}(\F(X)\cap \cP_{\disc})\subset \class_{r}(\F(X)\cap \cP(F_{y})) = \estclass_{\GLS}(X)\,\, \text{almost surely under }F_{y}.\]
    Then, for any $u\in \class_{r}(\F(X)\cap \cP_{\disc})$, there exists $u_{0, y}\in \estclass_{\GLS}(X)$, which might depend on $y$, such that
    \[u - u_{0, y} = 0 \,\,\text{almost surely under }F_{y}.\]
    Now we show that $u_{0, y}$ does not depend on $y$. For any $y_1\neq y_2\in \R^{n}$, since $\{e_{i}: i = 1, \ldots, n\}$ are contained in the support of both $F_{y_1}$ and $F_{y_2}$, 
    \[u(e_i) - u_{0, y_1}(e_i) = u(e_i) - u_{0, y_2}(e_i) = 0 \quad \text{for all }i = 1,\ldots,n.\]
    This implies that
    \[u_{0, y_1}(e_i) - u_{0, y_2}(e_i) = 0\quad \text{for all }i = 1,\ldots,n.\]
    Since $u_{0, y_1}(z) - u_{0, y_2}(z)$ is linear in $z$ and $\sp(\{e_1, \ldots, e_n\}) = \R^{n}$, we must have $u_{0, y_1} = u_{0, y_2}$. Thus, there exists $u_{0} \in \estclass_{\GLS}(X)$ such that, for any $y\in \R^{n}$,
    \[u - u_{0} = 0 \,\,\text{almost surely under }F_{y}.\]
    Since $y$ is in the support of $F_{y}$, this implies
    \[u(y) = u_{0}(y)\,\, \text{for all }y\in \R^{n}.\]
    Thus, $u\in \estclass_{\GLS}(X)$. 

    \subsection{Proof of Theorem \ref{thm:koopmann_master} (c)}\label{subapp:koopmann_master}
    We first prove the following lemma.
    \begin{lemma}\label{lem:basis}
      For any positive integer $n$, there exists a finite subset $\mathcal{Y}_0$ of $\R^{n}$ such that
      \begin{equation}
        \label{eq:span_condition}
        \sp\lb \com{z}{\vec^{*}(zz')}: z\in \mathcal{Y}_0\rb = \R^{(n+1)n / 2}.
      \end{equation}
    \end{lemma}
    \begin{proof}
      Let $e_{1}, \ldots, e_{n}$ be the canonical basis of $\R^{n}$ and $z_{ij} = e_{i} + e_{j}$ for any $1\le i\le j\le n$. We will prove that $\mathcal{Y}_0 = \{e_{i}: 1\le i\le n\}\cup \{z_{ij}: 1\le i\le j\le n\}$ satisfies \eqref{eq:span_condition}. Consider any $\{a_{i}: 1\le i\le n\}$ and $\{b_{ij}: 1\le i\le j\le n\}$ such that
      \[\sum_{i\le n}a_{i}\com{e_{i}}{\vec^{*}(e_{i}e_{i}')} + \sum_{i\le j\le n}b_{ij}\com{z_{ij}}{\vec^{*}(z_{ij}z_{ij}')} = \zero_{(n+1)n/2}.\]
      This can be decomposed into two equations:
      \begin{equation}
        \label{eq:span_1}
        \sum_{i\le n}a_{i}e_{i} + \sum_{i\le j\le n}b_{ij}(e_{i} + e_{j}) = \zero_{n},
      \end{equation}
      and
      \begin{equation}
        \label{eq:span_2}
        \sum_{i\le n}a_{i}e_{i}e_{i}' + \sum_{i\le j\le n}b_{ij}(e_{i} + e_{j})(e_{i} + e_{j})' = \zero_{n\times n}.
      \end{equation}
      For any $i < j$, the $(i, j)$-th entry of \eqref{eq:span_2} implies that
      \[b_{ij} = 0.\]
      Then \eqref{eq:span_1} and \eqref{eq:span_2} can be simplified into
      \[\sum_{i\le n}a_{i}e_{i} + 2\sum_{i\le n}b_{ii}e_{i} = \zero_{n}, \quad \sum_{i\le n}a_{i}e_{i}e_{i}' + 4\sum_{i\le n}b_{ii}e_{i}e_{i}' = \zero_{n\times n}.\]
      Thus, for any $i$,
      \[a_{i} + 2b_{ii} = a_{i} + 4b_{ii} = 0\Longrightarrow a_{i} = b_{ii} = 0.\]
      Therefore, the vectors in $\mathcal{Y}_0$ are linearly independent. Since $|\mathcal{Y}_0| = (n+1)n /2$, the lemma is proved.
    \end{proof}
    
    Next we will prove Theorem \ref{thm:koopmann_master}. Let $\mathcal{Y}_0$ be a subset satisfying \eqref{eq:span_condition}. For each $\beta\in \B\cup -\B$, where $-\B = \{-\beta: \beta\in \B\}$, let
    \[F_{0, \beta} = \frac{1}{2|\mathcal{Y}_0|}\sum_{z\in \mathcal{Y}_0}(\delta_{z} + \delta_{2X\beta - z}).\]
    For any $\Lambda\in \{\sigma_1^2 \Sigma, \sigma_2^2\Sigma\}$, since $\Lambda\succ 0$, there exists $c_{\beta, \Lambda} \in (0, 1)$ such that
    \[\Lambda \succ c_{\beta, \Lambda}\Cov_{F_{0, \beta}}[Y].\]
    Applying eigen-decomposition on $\Lambda - c_{\beta, \Lambda}\Cov_{F_{0, \beta}}[Y]$, there exists $w_{1, \beta, \Lambda}, \ldots, w_{n, \beta, \Lambda}\in \R^{n}$ such that
    \[\Lambda = c_{\beta, \Lambda}\Cov_{F_{0, \beta}}[Y] + \frac{1-c_{\beta, \Lambda}}{n}\sum_{i=1}^{n}w_{i, \beta, \Lambda} w_{i, \beta, \Lambda}'.\]
    Let
    \[F_{\beta, \Lambda} = c_{\beta, \Lambda}F_{0, \beta} + (1 - c_{\beta, \Lambda})F_{1, \beta, \Lambda}, \quad \text{where }F_{1, \beta, \Lambda} = \frac{1}{2n}\sum_{i=1}^{n}(\delta_{X\beta + w_{i, \beta, \Lambda}} + \delta_{X\beta - w_{i, \beta, \Lambda}}).\]
    Then
    \[\E_{F_{\beta, \Lambda}}[Y] = c_{\beta, \Lambda}\E_{F_{0, \beta}}[Y] + (1 - c_{\beta, \Lambda})\E_{F_{1, \beta}}[Y] = X\beta,\]
    and
    \[\Cov_{F_{\beta, \Lambda}}[Y] = c_{\beta, \Lambda}\Cov_{F_{0, \beta}}[Y] + \frac{1-c_{\beta, \Lambda}}{2n}\sum_{i=1}^{n}\{w_{i, \beta, \Lambda} w_{i, \beta, \Lambda}' + (-w_{i, \beta, \Lambda})(-w_{i, \beta, \Lambda})'\} = \Lambda.\]
    Thus,
    \begin{equation}
      \label{eq:F_beta_Lambda}
      F_{\beta, \Lambda}\in \F_{2r}(X; \beta, \Lambda)\cap \cP_{\disc}.
    \end{equation}
    Let
    \[F^{*} = \frac{1}{4|\B|}\sum_{\beta\in \B\cup -\B}\sum_{\Lambda\in \{\sigma_1^2\Sigma, \sigma_2^2\Sigma\}}F_{\beta, \Lambda}.\]
    Fix any $y\in \R^{n}$, $\beta^{*} = c_{11}\beta_1^{*}$ and $\Lambda^{*} = \sigma_1^2\Sigma$, let
    \[F_{0, y} = \frac{1}{2}F^{*} + \frac{1}{4}\delta_{y} + \frac{1}{4}\delta_{4X\beta^{*} - y}.\]
    Then
    \[\E_{F_{0, y}}[Y] = \frac{1}{2}\cdot 0 + \frac{1}{4}(y + 4X\beta^{*} - y) = X\beta^{*}.\]
    Using the same argument as above, there exists $c_y>0$ and $w_{1}, \ldots, w_{n}\in \R^{n}$ such that
    \[\Lambda^{*} = c_y\Cov_{F_{0, y}}[Y] + \frac{1-c_y}{n}\sum_{i=1}^{n}w_{i} w_{i}'.\]
    Let
    \[F_{y} = c_yF_{0, y} + (1 - c_y)F_{1, y}, \quad \text{where }F_{1, y} = \frac{1}{2n}\sum_{i=1}^{n}(\delta_{X\beta^{*} + w_{i}} + \delta_{X\beta^{*} - w_{i}}).\]
    Then
    \[\E_{F_y}[Y] = X\beta^{*}, \quad \Cov_{F_y}[Y] = \Lambda^{*}.\]
    As a result,
    \[F_y \in \F(X)\cap \cP_{\disc}.\]
    By \eqref{eq:F_beta_Lambda}, for any $\beta\in \B, \Lambda\in \{\sigma_1^2\Sigma, \sigma_2^2\Sigma\}$,
    \[F_{\beta, \Lambda}\in \cP(F_y)\Longrightarrow \F_{2r}(X; \beta, \Lambda)\cap \cP(F_y) \neq \emptyset.\]
    By Proposition \ref{prop:monotonicity} and the part (a),
    \[\class_{r}(\F(X)\cap \cP_{\disc})\subset \class_{r}(\F(X)\cap \cP(F_{y})) = \estclass_{\koopmann}^{\Sigma}(X)\,\, \text{almost surely under }F_{y}.\]
    Then, for any $u\in \class_{r}(\F(X)\cap \cP_{\disc})$, there exists $u_{0, y} \in \estclass_{\koopmann}^\Sigma(X)$, which might depend on $y$, such that
    \begin{equation}
      \label{eq:u-u0y}
      u - u_{0, y} = 0 \,\,\text{almost surely under }F_{y}.
    \end{equation}
    Following the same argument as in the proof of Theorem \ref{thm:GLS_master} (c), it remains to prove that $u_{0, y}$ does not depend on $y$. Assume that
    \[u_{0, y}(z) = a_y'z + z'B_y z = a_y'z + \vec^{*}(B_y)' \vec^{*}(zz') = (a_y', \vec^{*}(B_y)')\com{z}{\vec^{*}(zz')}.\]
    For any $y_1\neq y_2\in \R^{n}$, since $\mathcal{Y}_0$ is contained in the support of both $F_{y_1}$ and $F_{y_2}$, 
    \[((a_{y_1} - a_{y_2})', (\vec^{*}(B_{y_1}) - \vec^{*}(B_{y_2}))')\com{z}{\vec^{*}(zz')} = 0\quad \text{for all }z\in \mathcal{Y}_0.\]
 By \eqref{eq:span_condition}, $a_{y_1} = a_{y_2}$ and $\vec^{*}(B_{y_1}) = \vec^{*}(B_{y_2})$ and thus $u_{0, y_1} = u_{0, y_2}$. This proves that $u_{0, y}$ does not vary with $y$. Rewrite $u_{0, y}$ as $u_0$. For any $y\in \R^{n}$, \eqref{eq:u-u0y} implies that $u(y) = u_{0}(y)$ since $y$ is a point mass of $F_{y}$. The proof is then completed

 \subsection{Miscellaneous}
\begin{proposition}[\cite{conway2019course}, a special case of Theorem 6.13]\label{prop:hahn_banach}
  For any probability measure $F$, and subset $\G \subset \L^{1}(F)$, let $\G^{\perp}$ denote the orthogonal complement of $\G$ in $\L^{\infty}(F)$ (with respect to the canonical bilinear form that maps $(a, b) \in (\L^{1}(F) \times \L^{\infty}(F))$ to $\E_{F}[a(Y)b(Y)]$):
  \[\G^{\perp} = \left\{g'\in \L^{\infty}(F): \int gg' dF = 0\right\},\]
  and $(\G^{\perp})^{\perp}$ the orthogonal complement of $\G^{\perp}$ in $\L^{1}(F)$ defined similarly as above. Then
  \[(\G^{\perp})^{\perp} = \cl_{F}(\sp(\G)).\]
\end{proposition}


\begin{proposition}[\cite{conway2019course}, a special case of Proposition 3.3]\label{prop:finite_dim}
  For any probability measure $F$, and finite dimensional subspace $\G$ of $\L^{1}(F)$,
  \[\G = \cl_{F}(\G) \text{ almost surely under }F.\]
\end{proposition}

\end{document}